\documentclass[letterpaper, 10 pt, conference]{ieeeconf}  

\let\labelindent\relax
\usepackage{enumitem}
\usepackage{amssymb,amsmath,amsthm}
\usepackage{times}
\usepackage{todonotes}
\usepackage{multicol}
\usepackage{amsfonts}
\usepackage{booktabs}
\usepackage{multirow}
\usepackage[bookmarks=true,backref=page]{hyperref}
\usepackage{balance}
\usepackage[group-separator={,}]{siunitx}
\usepackage{gensymb} 
\usepackage{caption}
\usepackage[capitalise]{cleveref}
\definecolor{yellowcustom}{RGB}{200, 200, 0}
\definecolor{BrickRed}{RGB}{175, 50, 50}
\definecolor{Plum}{RGB}{142, 69, 133}

\newtheorem*{example}{Running example}
\theoremstyle{definition}
\newtheorem{proposition}{Proposition}

\newtheorem{lemma}{Lemma}
\newtheorem{theorem}{Theorem} 
\newtheorem{remark}{Remark}

\newtheorem{assumption}{Assumption} 

\usepackage{siunitx}

\usepackage{subcaption}

\usepackage[ruled,vlined,linesnumbered]{algorithm2e}
\DontPrintSemicolon
\SetAlgoSkip{}

\IEEEoverridecommandlockouts                              

\usepackage{bm}

\newcommand{\traj}{\bm{x}}
\newcommand{\ctrl}{\bm{u}}

\newcommand{\Lx}{L^{x}}
\newcommand{\Lout}{L^{\text{out}}}
\newcommand{\Lin}{L^{\text{in}}}
\newcommand{\Lio}{L^{\text{io}}}
\newcommand{\ball}{\text{B}}
\newcommand{\rect}{\text{Rect}}
\newcommand{\pibackup}{\pi^{\text{backup}}}
\newcommand{\pisafety}{\pi^{\text{safe}}}
\newcommand{\piexploration}{\pi^{\text{exp}}}
\newcommand{\pisafeexploration}{\pi}
\newcommand{\safetythreshold}{\epsilon}
\newcommand{\ntraj}{n^{\textnormal{traj}}}
\newcommand{\niter}{n^{\textnormal{iter}}}

\newcommand{\target}{l}
\newcommand{\unsafefcn}{g}
\newcommand{\unsafeset}{\mathcal{X}_{\mathcal{U}}}
\newcommand{\targetset}{\mathcal{X}_{\mathcal{T}}}

\newcommand{\norm}[1]{\left\lVert #1\right\rVert}

\newcommand{\dataset}{\mathcal{D}}

\newcommand{\safeset}{S}

\newcommand{\HRule}{\noindent\rule{\linewidth}{0.1mm}\newline}

\newcommand{\R}{\mathbb{R}}

\newcommand{\xdim}{n_x}
\newcommand{\udim}{n_u}
\newcommand{\uset}{U}
\newcommand{\datasetu}{\mathcal{D}_{u}}
\newcommand{\datasetv}{\mathcal{D}_{v}}
\newcommand{\tildeU}{\datasetu}
\newcommand{\ctrlsetv}{\datasetv}

\newcommand{\datadrivenV}{\widehat{V}}
\newcommand{\ddV}{\datadrivenV}

\newcommand{\tildeV}{\widetilde{V}}
\newcommand{\tildetraj}{\widetilde{\bm{x}}}
\newcommand{\tildectrl}{\widetilde{\bm{u}}}
\newcommand{\brt}{Avoid}
\newcommand{\reachavoidbrt}{ReachAvoid}

\newcommand{\ctrlv}{\bm{v}}

\newcommand{\ctrlw}{\bm{w}}
\newcommand{\hattraj}{\widehat{\bm{x}}}
\newcommand{\hatctrl}{\widehat{\bm{u}}}

\newcommand{\hatvi}{\widehat{v}^{i}}

\newcommand{\hatV}{\widehat{V}}

\newcommand{\truev}{\widetilde{v}^i}

\newcommand{\uncertainset}{\mathcal{E}}
\newcommand{\ddh}{\widehat{H}}
\DeclareMathOperator*{\argmax}{arg\,max}
\DeclareMathOperator*{\argmin}{arg\,min}

\newcommand{\airspeed}{\mathrm{v}}

\title{\LARGE \bf Data-Driven Hamiltonian for Direct Construction of \\ Safe Set from Trajectory Data}

\author{Jason J. Choi$^{*1}$, Christopher A. Strong$^{*1}$, Koushil Sreenath$^{1}$, Namhoon Cho\textsuperscript{\textdagger2}, and Claire J. Tomlin\textsuperscript{\textdagger1}
\thanks{This is the extended version of the article.}
\thanks{* Equal first authorship.\;\;\textdagger~Equal advising.}
\thanks{$^{1}$Jason J. Choi, Christopher A. Strong, Koushil Sreenath, and Claire J. Tomlin are with the University of California, Berkeley, CA 94720 USA.
        {\tt\small jason.choi, christopher\_strong@berkeley.edu}}%
\thanks{$^{2}$Namhoon Cho is with the Centre for Assured and Connected Autonomy at Cranfield University, United Kingdom.
}%
\thanks{This research is supported in part by the NASA ULI on Safe Aviation Autonomy, NSF Safe Learning-Enabled Systems, and the ONR LEARN projects. The work of Jason J. Choi received the support of a
fellowship from Kwanjeong Educational Foundation, Korea.}
}

\begin{document}

\maketitle
\thispagestyle{empty}
\pagestyle{empty}

\begin{abstract}
In continuous-time optimal control, evaluating the Hamiltonian requires solving a constrained optimization problem using the system's dynamics model. Hamilton-Jacobi reachability analysis for safety verification has demonstrated practical utility only when efficient evaluation of the Hamiltonian over a large state-time grid is possible. In this study, we introduce the concept of a data-driven Hamiltonian (DDH), which circumvents the need for an explicit dynamics model by relying only on mild prior knowledge (e.g., Lipschitz constants), thus enabling the construction of reachable sets directly from trajectory data. Recognizing that the Hamiltonian is the optimal inner product between a given costate and realizable state velocities, the DDH estimates the Hamiltonian using the worst-case realization of the velocity field based on the observed state trajectory data. This formulation ensures a conservative approximation of the true Hamiltonian for uncertain dynamics. The reachable set computed based on the DDH is also ensured to be a conservative approximation of the true reachable set. Next, we propose a data-efficient safe experiment framework for gradual expansion of safe sets using the DDH. This is achieved by iteratively conducting experiments within the computed data-driven safe set and updating the set using newly collected trajectory data. To demonstrate the capabilities of our approach, we showcase its effectiveness in safe flight envelope expansion for a tiltrotor vehicle transitioning from near-hover to forward flight.
\end{abstract}

\section{Introduction}

Data-driven safety verification is needed to ensure the safety of various real-world systems that involve uncertain dynamics. Such uncertainty, which may be impossible to model or require extensive modeling efforts, can result from various sources---complex aerodynamics in unconventional aircraft, robot manipulation of non-rigid objects, learning-enabled components in an autonomy stack, interaction with unstructured environments, and many others. 

Hamilton-Jacobi (HJ) reachability, rooted in model-based optimal control, provides a rigorous and flexible framework for verifying safety by computing the maximal safe set within constraints, as well as an associated safe policy \cite{bansal2017hamilton}. 
However, its practical applicability is limited by the need to evaluate the Hamiltonian, which is the optimal directional derivative of the value function along the dynamics. Evaluating the Hamiltonian involves solving a constrained optimization problem in the control input space. Consequently, both the lack of an accurate dynamics model as well as any computational challenges in the optimization can hinder the implementation of HJ reachability analysis in practice.

Existing data-driven approaches address these challenges \textit{indirectly} by using supervised learning to approximate components needed to evaluate the Hamiltonian. For instance, \cite{Fisac2018} employs Gaussian processes to learn the dynamics, while \cite{chilakamarri2024reachability} fits a neural network to approximate the Hamiltonian. However, these methods rely heavily on the quality of the learned model and may introduce inefficiencies due to misalignment between model learning and safety verification. Alternative data-driven methods outside the HJ framework primarily focus on forward reachability \cite{devonport2021data, djeumou2021fly, alanwar2023data, lew2021sampling}.

We propose a new data-driven approach by approximating the Hamiltonian used in HJ reachability \textit{directly} from data. The constrained optimization involved in computing the original Hamiltonian is replaced with an efficient data-enabled approximation, which computes its best approximation under the worst-case realization of the dynamics inferred from the data. Using this approximate Hamiltonian in HJ reachability allows the safe set to be computed directly from trajectory data without any intermediate supervised learning step.

Our main contributions are summarized below:
\begin{itemize}[leftmargin=1.25em, labelindent=\parindent, listparindent=\parindent, labelwidth=0em]
    \item We propose the \textit{data-driven Hamiltonian} (DDH), a novel concept enabling \textit{direct} data-driven reachability analysis for uncertain dynamics under mild assumptions of knowledge of the system's Lipschitz constants.
    \item We prove that using this DDH in HJ reachability computations guarantees a conservative approximation of the true safe set.
    \item We present a safe set expansion framework built upon the DDH-based reachability method, which can iteratively update a data-driven safe set while running experiments that safely collect more data.
    \item We demonstrate the effectiveness of our approach in safe longitudinal flight envelope expansion for a tiltrotor vehicle transitioning from near-hover to forward flight.
\end{itemize}

\textit{Notations.}
The norm $\norm{\cdot}$ being used is the $l_2$ norm, and $|\cdot|$ denotes an absolute value. An $l_2$ hypersphere centered at the origin with radius $r$ will be notated as $\ball(r) = \{x \mid \norm{x} \le r \}$. A hyperrectangle centered at the origin with element-wise radii $r\!\in\!\R^{\xdim}$ will be denoted as $\rect(r) \!=\!\{x \mid |x_j| \le r_j, \forall j\!=\!1, \cdots, \xdim \}$. The symbol $\oplus$ indicates the Minkowski sum of two sets. The superscript $i$ denotes the index of a data point and the subscript $j$ denotes the $j$-th element of a vector unless noted otherwise.

\section{Problem Formulation \& Background}

\subsection{Problem Formulation}
\label{subsec:problem}

We consider nonlinear system dynamics
\begin{equation}
    \dot{\traj}(s)\!=\!f(\traj(s), \ctrl(s)) \;\; \text{for} \; s\!\in\![-t, 0], \quad \traj(-t) = x,
\label{eq:dynsys}
\end{equation}
with state $\traj(s) \in \mathbb{R}^{\xdim}$, control $\ctrl(s) \in \uset \subset \mathbb{R}^{\udim}$, and initial state $x$ at time $-t$, where $t\!>\!0$ and $\uset$ is the control input set. The vector field $f$ is assumed to be Lipschitz continuous in the state, which is required for the forward completeness of the trajectory \cite{sastry2013nonlinear}.
We consider various forms of Lipschitz continuity:

\noindent (i) uniform Lipschitz constant $\Lx$, satisfying
\begin{equation}
    \norm{f(x, u) - f(x', u)} \le \Lx \norm{x - x'},
    \label{eq:lipschitz-uniform}
\end{equation}
\noindent (ii) input-element-wise constant vector $\Lin$ satisfying
\vspace{-0.5em}
\begin{equation}
    \norm{f(x, u) - f(x', u)} \le \sum_{j=1}^{\xdim} \Lin_j \lvert x_j - x_j' \rvert 
    \label{eq:lipschitz-input}
\vspace{-1em}
\end{equation}
(iii) output-element-wise constant vector $\Lout$ satisfying for $i\!=\!1,\!\cdots\!,\!\xdim,$     \vspace{-0.5em}
\begin{equation}
    \lvert f_i(x, u) - f_i(x', u) \rvert \le \Lout_i \norm{x - x'}     
    \label{eq:lipschitz-output}
    \vspace{-0.25em}
\end{equation}
\noindent (iv) $\xdim \!\times\!\xdim$ sensitivity matrix $\Lio$ satisfying for $i\!=\!1,\!\cdots\!,\!\xdim,$
\vspace{-0.5em}
\begin{equation}
    \lvert f_i(x, u) - f_i(x', u) \rvert \le \sum_{j=1}^{\xdim} L_{ij}^{\text{io}} \lvert x_j - x'_j \rvert
    \label{eq:lipschitz-matrix}
    \vspace{-1.0em}
\end{equation}
for all $x, x' \in \mathcal{X}, u \in \mathcal{U}.$

While the dynamics $f$ itself is deemed uncertain and unknown, we assume that (i) a dataset of trajectories from this uncertain system is given as $\dataset=\{(x^i, u^i, v^i)\}_{i=1}^N,$ where $v^i := f(x^i, u^i)$ denotes ``state velocity'', and (ii) we know at least one of the Lipschitz constants of $f$ above. The dataset is collected from experiments, where the state, control, and state velocity are sampled at various time steps of the trajectories of \eqref{eq:dynsys}.
The Lipschitz constants may come from prior knowledge or can be estimated directly from the dataset $\dataset$. While access to a valid constant indicates that our method is not completely model-free and requires basic system knowledge, the required modeling effort is significantly less than accurately characterizing the full dynamics $f$.

The safety specification is given as an unsafe region in the state space we want to avoid, denoted as $\unsafeset$. We consider two notions of a safe set:

\noindent \textit{1. Avoid Backward Reachable Tube (BRT):} The $\brt$ BRT is the set of initial states from which the system can avoid reaching the unsafe set $\unsafeset$ over the time horizon $t$:
\vspace{-0.5em}
\begin{align*}
        \brt(t; \unsafeset) \!= \!\{ & x\!\in\!\mathbb{R}^{\xdim} \mid \exists \ctrl(\cdot)\;\text{s.t.}\;\forall s\!\in\![-t, 0], \traj(s)\!\notin\!\unsafeset \}.
\vspace{-1em}
\end{align*}

\noindent \textit{2. Reach-Avoid BRT:} The $\reachavoidbrt$ BRT is the set of initial states from which the system can reach a target set, specified as $\targetset$, while avoiding the unsafe set $\unsafeset$:
\vspace{-0.5em}
\begin{align*}
    &\reachavoidbrt(t; \targetset, \unsafeset) \!= \!  \left\{x \in \mathbb{R}^{\xdim} \mid \exists \ctrl(\cdot) \; \text{s.t.} \right. \\ 
    &\qquad \left. \exists s \in [-t, 0], \traj(s) \in \targetset \; \& \; \forall \tau \in [-t, s], \traj(\tau) \notin \unsafeset  \right\}.
\vspace{-0.5em}
\end{align*} 
Here, $s$ corresponds to the time at which the system reaches $\targetset$, and $\tau$ indexes over the previous times to ensure that the system does not enter the unsafe set before reaching $\targetset$.

Our goal is to answer two questions: \textbf{(i) Safe set construction from data}: how can we estimate these safe sets directly from the trajectory data of the uncertain dynamical system? \textbf{(ii) Safe experiment design}: how can we design a sequence of safe experiments that gathers data safely in order to expand the safe set?

\subsection{Hamilton-Jacobi Reachability}
\label{subsec:HJR}
HJ reachability encodes the safety problem by first defining a value function whose sign indicates which states are included in the reachable set, and then uses dynamic programming to compute this value function. In this work, we focus on solving the two backward reachability problems described in Section \ref{subsec:problem}, and other standard reachability problems are detailed in \cite{herbert2020safe}. We first represent the unsafe and target sets as level sets of Lipschitz continuous functions, $\unsafefcn(\cdot)$ and $\target(\cdot)$, such that 
\begin{equation}
    \unsafeset = \{x \mid \unsafefcn(x) \le 0 \}, \;\; \targetset = \{x \mid \target(x) \ge 0 \}.
    \label{eq:unsafe-set}
\end{equation}

\noindent Then, we can define the value functions whose zero-superlevel sets represent the $\brt$ and $\reachavoidbrt$ BRTs:

\HRule
Value function for $\brt$ BRT:
\begin{align}
    \label{eq:brt-value}
    V(x, t)  :&\!\!= \sup_{\ctrl(\cdot)} \min_{s \in [-t, 0]}    \unsafefcn(\traj(s))  \\ 
    & \Rightarrow \brt(t; \unsafeset) =\{x \mid V(x, t) \ge 0 \}. \quad\quad\quad\quad \nonumber
\end{align}
Value function for $\reachavoidbrt$ BRT:
\begin{align}
    \label{eq:brat-value}
    V(x, t) :&\!\!= \sup_{\ctrl(\cdot)} \max_{s \in [-t, 0]} \min \!\left\{ \target(\traj(s)), \min_{\tau \in [-t, s] } \unsafefcn(\traj(\tau))\! \right\}\!\!\! \\
    & \Rightarrow \reachavoidbrt(t; \targetset, \unsafeset) =\{x \mid V(x, t) \ge 0 \}. \nonumber
\end{align}
\hrule
\vspace{0.75em}

\noindent The optimization problem in \eqref{eq:brt-value} seeks a control that maximizes the closest distance to the unsafe set boundary, and in \eqref{eq:brat-value}, a control that minimizes the distance to the target set while avoiding $\unsafeset$.

By applying the dynamic programming principle, the value functions become the viscosity solutions \cite{bardi1997optimal} to the following HJ partial differential equations (PDEs) that are in the variational inequality (VI) form \cite{fisac2015reach}:

\HRule
HJ-VI for $\brt$ BRT:
\begin{equation}
\label{eq:brt-vi}
    0\!=\!\min \left\{\unsafefcn(x)\!- \!V(x, t),\; -D_t V(x, t)\!+\!H(x, D_x V(x, t)) \right\},\;\;\;\;\;
\vspace{-0.5em}
\end{equation}
with terminal condition $V(x, 0) = \unsafefcn(x)$.
\vspace{1em}

\noindent HJ-VI for $\reachavoidbrt$ BRT:
\begin{align}
\label{eq:brat-vi}
    0\!=\!& \min \Big\{\unsafefcn(x)\! -\! V(x, t), \\ 
    & \;\max \{ \target(x)\!-\!V(x, t),  -D_t V(x, t)\!+\!H(x, D_x V(x, t)) \} \Big\}, \nonumber
\vspace{-0.25em}
\end{align}
with terminal condition $V(x, 0) = \min\{\target(x), \unsafefcn(x)\}$.
\vspace{0.5em}
\hrule
\vspace{0.75em}
\noindent In \eqref{eq:brt-vi} and \eqref{eq:brat-vi}, the \textit{Hamiltonian} $H(x, p)$ is defined as  
\begin{equation}
    H(x, p) = \max_{u \in U} p^\top f(x, u)
    \label{eq:hamiltonian}.
\end{equation}
To solve \eqref{eq:hamiltonian} directly, we require knowledge of the dynamics $f$. Even when $f$ is known, \eqref{eq:hamiltonian} may be a nonconvex optimization problem, without further restrictive assumptions such as the control-affineness of $f$ and the convexity of $U$.

\begin{remark}
\label{remark:infinite horizon} The BRTs guarantee the safety constraint $\traj(s)\!\notin\! \unsafeset$ only for a finite horizon. Two measures can be taken to guarantee safety for an indefinite horizon. First, we can compute the $\brt$ BRT for a sufficiently long horizon until the BRT converges to the maximal control invariant set in $\unsafeset^c$. However, to avoid the issue of discontinuity or non-uniqueness of the HJ-VI solution, a discount factor must be introduced to the value function \cite{akametalu2018minimum, xue2018reach, choi2023forward}. An alternative approach is to design the target set in the $\reachavoidbrt$ BRT as a control invariant set, which results in the finite-time $\reachavoidbrt$ BRT also being control invariant. We employ the second approach in \cref{sec:experiment_design} for the iterative safe set expansion algorithm.
\end{remark}

\section{Data-driven Hamiltonian}
\label{sec:ddh}
Our goal is to find a data-driven estimate of the true reachable sets while providing rigorous guarantees on the states included within these sets. Towards this end, we first present the concept of the \textit{data-driven Hamiltonian (DDH)}, which is a lower bound of the Hamiltonian in \eqref{eq:hamiltonian} constructed using the collected trajectory data.
We then prove that computing the value function based on DDH yields a conservative estimate (inner-approximation) of the BRTs.

\subsection{Concept}
\label{subsec:vfb-rep}
The general idea in our approach is to represent the explicit dynamics abstractly as a state velocity vector $v:=f(x, u)$, and adopt a geometric viewpoint of the reachability problem.
We define the \textit{vector field bound} (VFB) as the set of possible velocities at a given state $x \in \R^{\xdim}$:
\begin{equation}
\label{eq:vfb}
    F(x) = \{ f(x, u)\;|\; u \in U\}.
\end{equation}
Under this abstraction, the dynamics in \eqref{eq:dynsys} can be equivalently represented as the differential inclusion \cite{goebel2009hybrid},
\begin{equation}
    \dot{\traj}(s) \in F(\traj(s)),
    \label{eq:dynsys_cbf}
\end{equation}
and the Hamiltonian in \eqref{eq:hamiltonian} can be written as
\begin{equation}
    H(x, p) := \max_{v \in F(x)} p^\top v,
    \label{eq:hamiltonian-vfb}
\end{equation}
where the objective function in \eqref{eq:hamiltonian-vfb} becomes a linear objective in $v$. We denote $v^*:=\argmax_{v \in F(x)} p^\top v$.

Next, given a single observation in our dataset $\dataset$, $(x^i, u^i, v^i)$, we reason about what information we have at the state $x$. Let the \textit{true velocity} at the state $x$ resulting from the observed control $u^i \in U$ be
\[\truev := f(x, u^i) \; \; \in\!F(x).
\]
Since we have observed $v^i = f(x^i, u^i)$, we can construct an uncertain estimate of $\truev$ around $v^i$ by considering the notion of an \textit{uncertainty set} $\uncertainset(x; x^i)$. This set bounds how much the velocity could have changed between $x^i$ and $x$ 
to satisfy the following requirement:
\begin{assumption}[Valid Uncertainty Sets]
\label{assumption:uncertainty}
The uncertainty set, represented as a set-valued map $\uncertainset:\R^{\xdim} \times \R^{\xdim} \rightarrow 2^{\R^{\xdim}}$, whose output is a closed set in the velocity space, satisfies
\begin{equation}
\label{eq:uncertainty-set}
\truev - v^i \in \uncertainset(x; x^i)   
\end{equation}
or equivalently,
    $\truev \in v^i \oplus \uncertainset(x; x^i)$,
for all $(x^i, u^i, v^i) \in \dataset$ and $x \in \R^{\xdim}$.
\end{assumption}

Under this assumption, we construct a lower bound on the Hamiltonian by evaluating the minimum of the linear objective $p^\top v$ over the uncertainty set around each observation. This yields our proposed \textit{DDH}:

\HRule
\textbf{Data-driven Hamiltonian (DDH):}
\begin{equation}
\label{eq:ddh_general}
    \widehat{H}(x, p) := \max_{i \in \{1, \cdots, N\}} \min_{\hatvi \in v^i \oplus \uncertainset(x; x^i)} p^\top \hatvi
\end{equation}
\hrule
\vspace{0.75em}
\noindent 
The $\min$ operation considers the worst-case realization of the uncertainty associated with each data point and the $\max$ operation reasons about what data point provides the best estimate of the Hamiltonian despite the uncertainty.

Notice that for all $i\!\in\!\{1, \cdots, N\}$, since $\truev \in v^i \oplus \uncertainset(x; x^i)$ and also $\truev \in F(x)$,
\[
    \min_{\hatvi \in v^i \oplus \uncertainset(x; x^i)} p^\top \hatvi \le p^\top \truev \le \max_{v \in F(x)} p^\top v = H(x, p).
\]
This yields the following proposition:

\begin{proposition}
\label{thm:ddh lower bound}
If $\uncertainset$ satisfies Assumption \ref{assumption:uncertainty}, the DDH is a guaranteed lower bound of the true Hamiltonian:
\begin{equation}
\label{eq:ham_comparison} 
    \widehat{H}(x, p) \le H(x, p).
\end{equation}
\end{proposition}

\cref{fig:diagram} provides a visual explanation of this lower bound mechanism.

\begin{figure}
\centering
\includegraphics[width=\columnwidth]{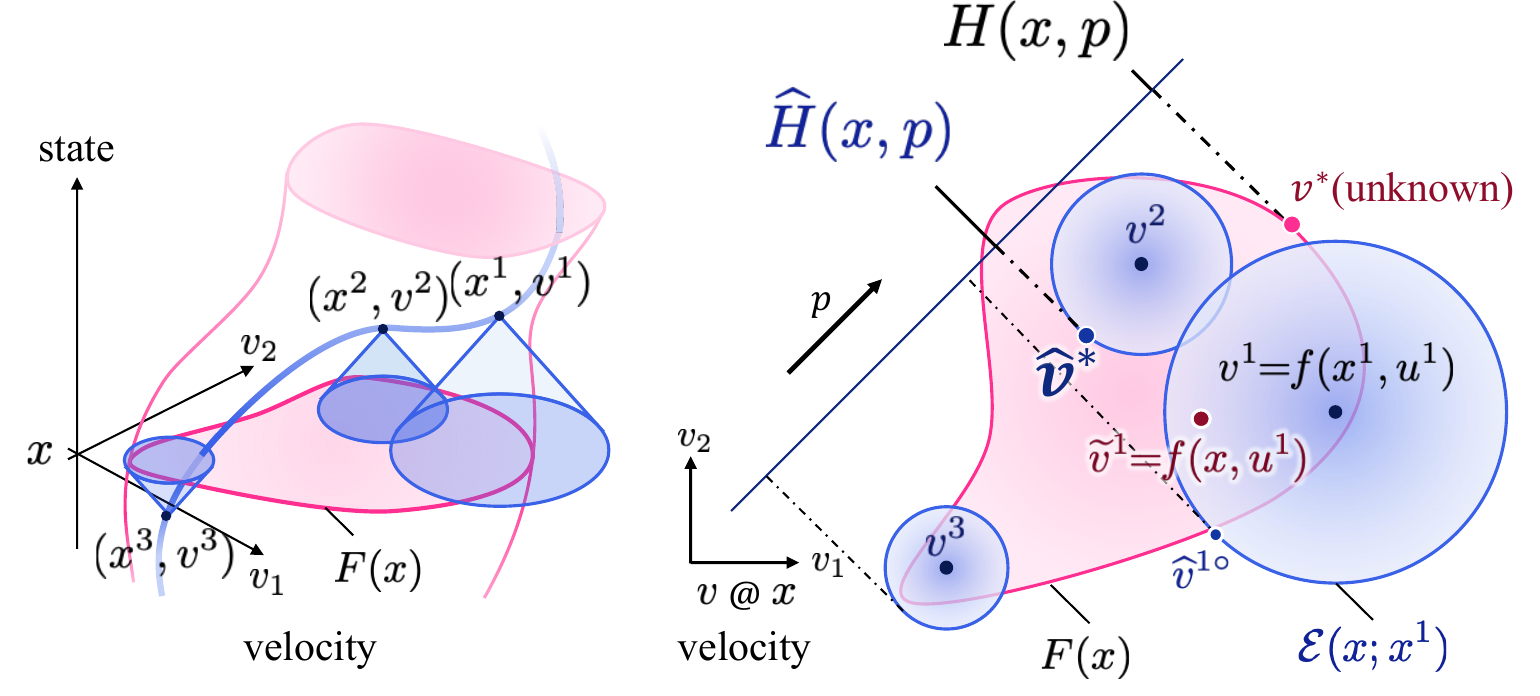}
\vspace{-1em}
\caption{
Illustration of Data-driven Hamiltonian. (left) Velocity space indexed by state. The trajectory data consist of state velocities $v^{i}$ indexed by states $x^{i}$ (\textcolor{blue}{blue}). The true VFB at a query state $x$, $F(x)$, is unknown (\textcolor{magenta}{red}). We can estimate $F(x)$ by mapping data $v^{i}$ at $x^{i}$ to true velocity $\widetilde{v}^{i}$ at $x$. (right) Velocity space at $x$ (top-down view of the left). $\widetilde{v}^{i}$ lies in an uncertainty set $\uncertainset(x;x^{i})$ propagated from $x^{i}$ to $x$ (\textcolor{blue}{blue circle}). Given the costate $p$, the DDH $\widehat{H}(x,p)$ in \eqref{eq:ddh_general} takes the best guess $\widehat{v}^{*}$ among $\widehat{v}^{i\circ}$'s, the worst-case realization of $\widetilde{v}^{i} \in \uncertainset(x;x^{i})$. This procedure ensures $\widehat{H}(x,p) \leq H(x,p)$.
}
\label{fig:diagram}
\end{figure}

\subsection{Practical Implementation}
\label{subsec:ddh-implementation}

In this section, we describe several instantiations of the uncertainty set and the resulting DDHs by using various levels of system knowledge described in \cref{subsec:problem}. In each case, the DDH is computed by $i^* = \argmax_{i = \{1, \hdots, N\}} p^\top \widehat{v}^{i\circ}$ where $\widehat{v}^{i\circ} = \argmin_{\hatvi \in v^i \oplus \uncertainset(x; x^i)} p^\top \hatvi$.
\subsubsection{$l_2$-ball DDH}
Here we consider the uncertainty sets we can obtain from knowing the uniform Lipschitz constant $\Lx$ or the input-element-wise Lipschitz constant vector $\Lin$. With $\Lx$, we can use $\uncertainset_{\Lx}(x; x^i):=\ball(\Lx \norm{x - x^i})$, and with $\Lin$, we can use
$\uncertainset_{\Lin}(x; x^i):=\ball(\Lin{}^\top \lvert x - x^i \rvert),$
which guarantees \cref{assumption:uncertainty} based on \eqref{eq:lipschitz-uniform} and \eqref{eq:lipschitz-input}, respectively. Since minimizing a linear objective over an $l_2$-ball has a closed-form solution, the $l_2$-ball DDH can be found with
\begin{equation}
\label{eq:opt_ham_sphere}
    \begin{aligned}
        \widehat{v}^{i\circ} = v^i - r^i (x) \frac{p}{\norm{p}}, 
    \end{aligned}
\end{equation}
where $r^i (x)$ corresponds to the radius of $\uncertainset_{\Lx}(x; x^i)$ or $\uncertainset_{\Lin}(x; x^i)$ respectively. 
\subsubsection{Hyperrectangle DDH}
Here we consider the uncertainty sets we can obtain from knowing the output-element-wise Lipschitz constant vector $\Lout$ or the sensitivity matrix $\Lio$. With $\Lout$, we can use $\uncertainset_{\Lout}(x; x^i):=\rect(\Lout \norm{x - x^i})$,
and with $\Lio$, we can use $\uncertainset_{\Lio}(x; x^i) := \rect(\Lio \lvert x - x^i \rvert)$. Since minimizing a linear objective over a hyperrectangle has a closed-form solution, both element-wise DDHs can be found with
\begin{equation}
\begin{aligned}
\vspace{-1em}
    &\widehat{v}^{i\circ}_{j} = \begin{cases}
                        v_{j}^{i} - r^i_j(x) & \text{if } p_j \ge 0, \\
                        v_{j}^{i} + r^i_j(x) & \text{otherwise},
                        \end{cases} 
\end{aligned}
\end{equation}
where $r^i(x)$ corresponds to the radii of $\uncertainset_{\Lout}(x; x^i)$ and $\uncertainset_{\Lio}(x; x^i)$ respectively. 

\begin{remark} (Computational cost)
Fast computation of the DDH is essential since it must be conducted at all state-time grid points for the HJ reachability computation. The DDH computation scales linearly with the number of data points, as it involves solving the inner optimization problem independently for each data point and taking the maximum. The cost of the inner optimization for both the $l_2$-ball and hyperrectangle DDH scales linearly with the state dimension, as computing the uncertainty set and evaluating the optimization problem are both linear in the state dimension. Therefore, the DDH at a given state $x$ can be computed in $O(N \xdim)$ operations.
\end{remark}

\subsection{Additional System Knowledge}
\label{subsec:additional}
We propose additional modifications to the DDH in \eqref{eq:ddh_general} for when we have further information about the system in order to reduce the gap between $\ddh$ and the true $H$. For example, physical systems are subject to a reasonable range of state velocities, and we can use this velocity bound to refine our DDH. Suppose we know that a set-valued map $G(x)$ bounds the VFB, satisfying $F(x) \subseteq G(x)$,
for all states $x$ within the computation domain.
Then, 
\[
   \min_{v \in G(x)} p^\top v \le \min_{v \in F(x)} p^\top v \le \max_{v \in F(x)} p^\top v = H(x, p).
\] 
As a result, we can improve the DDH to $\ddh_{G}(x, p)$ where
\[
    \resizebox{.99\hsize}{!}{$\displaystyle
    \ddh(x, p) \!\le\! \ddh_{G}(x, p)\!:=\!\max\!\left\{\min_{v \in G(x)} p^\top v, \ddh(x, p)\right\} \!\le \!H(x, p)
    $}.
\]
If $G(x)$ consists of simple shapes like hyperspheres or hyperrectangles, $\min_{v \in G(x)} p^\top v$ is computationally inexpensive to evaluate. Using $\ddh_{G}$ can be an efficient way to reduce the approximation error of our approach in regions where data is scarce and as a result $\uncertainset$ is large. 

Finally, the DDH can be applied modularly if only a partial component of the system dynamics is uncertain. For instance, if the dynamics consists of two subsystems,
\begin{equation*}
    \begin{bmatrix} \dot{\traj}_1(s) \\ \dot{\traj}_2(s)
    \end{bmatrix} =     \begin{bmatrix} f_1(\traj_1(s), \traj_2(s), u_1(s)) \\  f_2(\traj_1(s), \traj_2(s), u_2(s))
    \end{bmatrix},
\end{equation*}
where $f_1$ is unknown and $f_2$ is known, the Hamiltonian can be decomposed into $H(x, p) = \max_{u_1} p_1^\top f_1(x_1, x_2, u_1) + \max_{u_2} p_2^\top f_2(x_1, x_2, u_2)$. We can replace only the first term with the DDH and keep the second term, which can be computed based on the known model of $f_2$.

\subsection{Data-driven Value Functions \& Safe Sets}

Define the \textit{data-driven value functions}, $\widehat{V}(x, t)$, as the solution of the HJ-VIs in \eqref{eq:brt-vi} and \eqref{eq:brat-vi}, where $H$ is replaced with the DDH $\ddh$ in \eqref{eq:ddh_general}. For instance, $\widehat{V}$ for the $\brt$ BRT is defined as the solution to

\HRule
DDH-based HJ-VI for $\brt$ BRT:
\begin{equation}
\label{eq:hj_vi_data}
    0\!=\!\min \left\{\unsafefcn(x)\!- \!\ddV(x, t),\; -D_t \ddV(x, t)\!+\!\ddh(x, D_x \ddV(x, t)) \right\},\;\;\;\;\;
\vspace{-0.5em}
\end{equation}
with terminal condition $\ddV(x, 0) = \unsafefcn(x)$.
\vspace{0.5em}
\hrule
\vspace{0.75em}

\noindent We can similarly define the data-driven value function for the $\reachavoidbrt$ BRT. Presented next is the main theoretical result of this paper.

\begin{figure*}[t!]
\centering
\includegraphics[width=\textwidth]{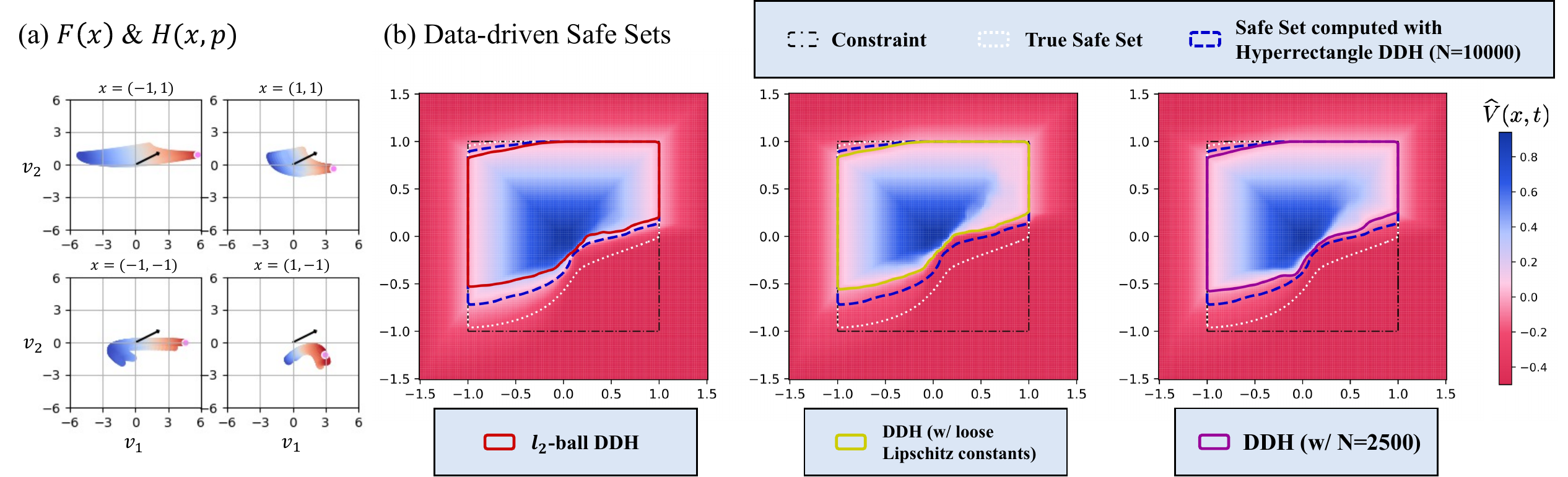}
\vspace{-1.25em}
\caption{Random Polynomial Systems. (a) Example of the vector field bound $F(x)$ at various states and $\argmax_{v \in F(x)} p^\top v$, illustrating the nonconvexity of the optimization. (b) Data-driven safe sets ($\brt$ BRT) computed using the DDH: (1) hyperrectangle DDH using tight sensitivity matrix $\Lio$ (\textcolor{blue}{blue}), (2) $l_2$-ball DDH based on Lipschitz constant $\Lx$ (\textcolor{BrickRed}{red}), (3) hyperrectangle DDH with doubled matrix $\Lio$ (\textcolor{yellowcustom}{yellow}), and (4) hyperrectangle DDH with fewer data points (\textcolor{Plum}{purple}).}
\label{fig:running-example}
\end{figure*}

\begin{theorem}
\label{thm:main}
The data-driven value function $\datadrivenV(x, t)$ is a \textit{guaranteed lower bound} of the true value function of the BRT problems in \eqref{eq:brt-value} and \eqref{eq:brat-value}:
\begin{equation}
\label{eq:value_comparison}
    \datadrivenV(x, t) \le V(x, t).
\end{equation}
\end{theorem}
\begin{proof}
The main idea of the proof is to reveal the \textit{inverse optimality of the DDH}. We define a fictitious dynamics that captures a differential game between the leader that selects its action among the observed velocities in $\dataset$, and an adversarial follower that selects its action as the worst-case realization of the uncertain dynamics. Then, we prove that the value function of this game is the viscosity solution of the DDH-HJ-VI \eqref{eq:hj_vi_data}, $\datadrivenV$. Finally, we show that the true dynamics is always the outcome of a less adversarial follower strategy. See Appendix for the full proof.
\end{proof}

\begin{theorem} We define the safe set $\safeset(t):=\{x\; | \; V(x, t) \ge 0\}$ and the data-driven safe set $\widehat{\safeset}(t):=\{x\; | \; \widehat{V}(x, t) \ge 0\}$. The safe set $\safeset(t)$ can be either $\brt(t; \unsafeset)$ or $\reachavoidbrt(t; \targetset, \unsafeset)$. Then, $\widehat{\safeset}(t)$ is a guaranteed inner-approximation of $\safeset(t)$:
\begin{equation}
    \widehat{\safeset}(t) \subseteq \safeset(t).
\end{equation}
\begin{proof}
    This is a direct result from \eqref{eq:value_comparison}.
\end{proof}
\end{theorem}

Finally, the data-driven value function can also provide a safe control policy within the computed safe set $\widehat{\safeset}(t)$. Define
\begin{equation}
\pi_{\hatV}(x, -t) = u^{i*},
\label{eq:safety-policy-k}
\end{equation}
where 
\[
i^* = \argmax_{i \in \{1, \cdots, N\}} \min_{\hatvi \in v^i \oplus \uncertainset(x; x^i)} D_x \hatV(x, t)^\top \hatvi
\]
is the solution of the optimization in \eqref{eq:ddh_general}. When we try to maximize $\hatV$, if we take $u^{i*}$ as our control, we are guaranteed to do better than the worst-case instantiation of the uncertainty. That is,
\begin{equation*}
    \widehat{H}(x, D_x \hatV(x, t)) \le  D_x 
 \hatV(x, t)^\top f(x, u^{i*}).
\end{equation*}
This means that $\pi_{\hatV}(\traj(s), s)$ can maintain safety, $\traj(s) \notin \unsafeset$ for $s \in [-t, 0]$. For the reach-avoid problem, we are guaranteed to reach the target set no later than $t$ under $\pi_{\hatV}$ for any initial state in $\reachavoidbrt(t; \targetset, \unsafeset)$.

\subsection{Running Example: Random Polynomial Systems}
We consider a system with $\xdim =\udim =2$ with dynamics
$f_i(x) = p_i(u_1, u_2) + q_i(u_1, u_2) x_1 + r_i(u_1, u_2) x_2$, for $i=1, 2$, where $p_i, q_i, r_i$ are quadratic polynomials in $u$ with randomly generated coefficients in $[-2, 2]$. The constant terms of $p_i$ are set to zero so that $x\!=\!0$ is an equilibrium under $u\!=\!0$. The randomness in the coefficients yields various nonconvex shapes for the VFB $F(x)$ at each state. Using the first-order optimality condition, we analytically solve for the true Hamiltonian (Fig. \ref{fig:running-example} (a)), from which we compute a true safe set to validate our DDH-based safe sets. Tight Lipschitz constants can also be computed analytically from the coefficients. In applying our methods, we assume no explicit knowledge of the dynamics, treating them as a black-box model. The dataset $\dataset$ is obtained by uniform sampling across the state domain and control bounds. In practice, collecting such data safely is challenging; we address realistic safe data collection further in the next section.

Fig. \ref{fig:running-example} (b) shows safe sets computed using the $\brt$ BRT formulation with four DDH instantiations. In all cases, the resulting sets correctly under-approximate the true safe set. Tighter Lipschitz bounds and larger datasets generally produce less conservative safe sets. However, not only the amount but also the informational content of the data significantly influences the safe set construction. Careful balancing between required prior information (e.g., Lipschitz bounds), data quantity and quality, and the resulting conservatism remains an important direction for future investigation.

\section{Iterative Safe Set Expansion}
\label{sec:experiment_design}

In this section, we introduce an algorithm for iteratively updating a data-driven $\reachavoidbrt$ BRT while maintaining safety throughout the data collection process. 
Our algorithm is motivated by the test procedure of real-world systems for safety verification---such as flight envelope expansion, which will be illustrated in \cref{sec:tiltrotor flight envelope}---where experiments start from an initial conservative safe region, designed with, for instance, local linear analysis. As such, we make the following assumption:

\begin{assumption}
\label{assumption:exp design initial region}
    The target set $\targetset$ must be \textit{forward invariant} under a \textit{backup policy} $\pibackup(x)$ and must not intersect with the unsafe set, $\targetset \cap \unsafeset = \emptyset$.
\end{assumption}

The target set is by definition contained in the safe set (the $\reachavoidbrt$ BRT) even when no data is collected.  Thus, we can safely initiate the data collection by setting the initial estimate of the safe set $\widehat{\safeset}_0$ as $\targetset$ and using $\pibackup$ for the initial experiments. 

\begin{algorithm}[t!]
\caption{Safe experiments for safe set expansion}
\label{alg:experiment-design}
\vspace{0.5em}
\KwIn{$\targetset=\{l(x)\ge0\}$: Control invariant target set, \\ 
$\unsafeset=\{g(x)\le0\}$: Unsafe set, \\
$\pibackup$: Backup policy, $\;\;\piexploration$: Exploration policy, \\
$\niter$: Number of iterations, \\
$\ntraj$: Number of experiments per iteration, \\
$\text{T}$: Time length of each rollout,  $\Delta t$: sampling time,\\
$\mathsf{GetInit}(\safeset; \ntraj, \dataset, V)$: Selects initial states within the safe set $\safeset$, given the dataset $\dataset$ and value function $V$, \\
$\mathsf{Rollout}(\pi; x_0)$: Obtains trajectory under policy $\pi$ by running experiment with initial state $x_0$,\\
$\mathsf{PruneData}(\dataset, V)$: Conducts data reduction  based on the value function $V$.
}
\KwOut{Final safe set $\widehat{\safeset}_{\niter}$}
\SetAlgoNoLine
Initialization: $\dataset_{0} = \{\}$, $\widehat{\safeset}_0\gets\targetset$, $\pisafety_{0}\gets\pibackup$ \;
\For{$k \gets 0$ \KwTo $\niter$}{
    $\{ x_0^{j} \}_{j=1}^{\ntraj}\!\gets\! \mathsf{GetInit}(\widehat{\safeset}_k; \ntraj, \dataset_{k}, \hatV_k)$ \;
     $\dataset_{k+1}\gets \dataset_{k}$ \;
    \For{$j \gets 1$ \KwTo $\ntraj$}{
        $\pi_{k}(x) \gets$ \eqref{eq:safe_exploration_policy} based on $\pibackup, \pisafety_{k}, \piexploration$.\;
        $\{x^i, u^i, v^i\}_{i=0}^{\lfloor \text{T} / \Delta t \rfloor} \gets \mathsf{Rollout}(\pi_{k}; x_0^j)$ \;
        $\dataset_{k+1} \gets \dataset_{k+1} \cup \{x^i, u^i, v^i\}_{i=0}^{\lfloor \text{T} / \Delta t \rfloor}$ \;
    }
    $\hatV_{k+1} \gets$ Compute \eqref{eq:hj_vi_data} with $\dataset_{k+1}$ \;
    $\widehat{\safeset}_{k+1}, \pisafety_{k+1} \gets$ Update from $\hatV_{k+1}, \dataset_{k+1}$ \;
    $\dataset_{k+1}\gets \mathsf{PruneData}(\dataset_{k+1}, \hatV_{k+1})$
}
\end{algorithm}

\textit{Overview} (Algorithm \ref{alg:experiment-design}). At each iteration $k$, we start the procedure by sampling initial conditions $\{x_0^j\}_{j=1}^{\ntraj}$ from the current data-driven safe set $\widehat{\safeset}_{k}$. We then roll out trajectories from each initial condition using a policy that is guaranteed to maintain the trajectory within $\widehat{\safeset}_k$, denoted by $\pisafeexploration_k$. The data from these rollouts are appended to the current dataset, yielding a new dataset $\dataset_{k+1}$ which is then used to compute the next iteration's data-driven safe set $\widehat{\safeset}_{k+1}$. A data reduction step is applied at the end of each iteration to reduce the computational expense of the process and retain only the data most relevant to safety. 

\textit{Initial States.}
The initial states are sampled near the boundary of the data-driven safe set $\widehat{\safeset}_{k}$.
Sampling initial states close to the boundary reduces the conservatism of the DDH at the boundary, thus helping to expand the safe set. The existing data in $\dataset$ can also guide the selection of the initial states, for example by encouraging the selection of states in regions with less data.

\textit{Safe Exploration Policy.} The safe exploration policy $\pisafeexploration_{k}$ is defined as
\begin{equation}
\pi_k(x) = \begin{cases}
            \pibackup(x)         &\text{if } \target(x) \ge 0, \\
            \pisafety_k(x)         &\text{if } \target(x) < 0 \;\&\; \hatV_{k}(x, t) < \safetythreshold,   \\ 
            \piexploration(x) &\text{otherwise,}    
\end{cases}
\label{eq:safe_exploration_policy}
\end{equation}
with safety threshold $\safetythreshold > 0$, where $t$ is the time horizon of the $\reachavoidbrt$ BRT. If the system is in the target set, we apply the backup controller, which is guaranteed to keep it within the set by \cref{assumption:exp design initial region}. If the state is outside the target set and away from the safe set boundary ($\hatV_{k}(x, t) \ge \safetythreshold$), we apply an exploration policy $\piexploration$ which does not need to satisfy any particular safety constraint. Finally, when the state is close to exiting the safe set ($\hatV_{k}(x, t) < \safetythreshold$), we apply the safety controller of the current data-driven safe set, $\pisafety_k(x)=\pi_{\hatV_{k}}(x, -t)$, where $\pi_{\hatV_{k}}$ is defined in \eqref{eq:safety-policy-k}. Regardless of the exploration policy $\piexploration$, the system under this switching safety filter can stay within $\widehat{\safeset}_k$ due to its control invariance noted in \cref{remark:infinite horizon} \cite{wabersich}.

\begin{proposition}
\label{thm:safe exploration}
If $\widehat{\safeset}_0 = \targetset$ and $\pibackup$ satisfy \cref{assumption:exp design initial region}, the trajectories in Algorithm \ref{alg:experiment-design} never enter $\unsafeset$.
\end{proposition}
\begin{proof}
Consider an arbitrary iteration $k$, with current data-driven safe set $\widehat{\safeset}_{k}$ from dataset $\dataset_{k}$. The initial state is contained in $\widehat{\safeset}_{k}$. As a result, since $\pisafeexploration_{k}$ renders $\widehat{\safeset}_k$ forward invariant, the resulting trajectory will be contained in $\widehat{\safeset}_k$. Since $\widehat{\safeset}_k$ is a $\reachavoidbrt$ BRT, $\widehat{\safeset}_k \cap \unsafeset = \emptyset$.   
\end{proof}

\textit{Data Reduction.}
We apply a data reduction that removes data points irrelevant to safety. 
This consists of keeping only the data whose indices show up as the optimal $i^*$ for the DDH evaluated during the computation of the value function $\hatV_{k}$. These data points are the ones actually used in the computation of $\hatV_{k}$ and the resulting data-driven safe control.
This reduction significantly helps the computation time with negligible impact on the resulting safe set.

\textit{Design Decisions.} A variety of design choices influence the outcome of the algorithm. In particular, the sampling strategy for initial states and the exploration policy $\piexploration$ have been empirically shown to be decisive factors in the efficiency of the safe set expansion. Additional parameters, including $\ntraj$ and the experiment length $\text{T}$, also affect the result. A more in-depth investigation into the varying effects of these design choices is left for future work.

\section{Application to Tiltrotor Safety Verification}
\label{sec:tiltrotor flight envelope}

\begin{figure}
\centering
\includegraphics[width=\columnwidth]{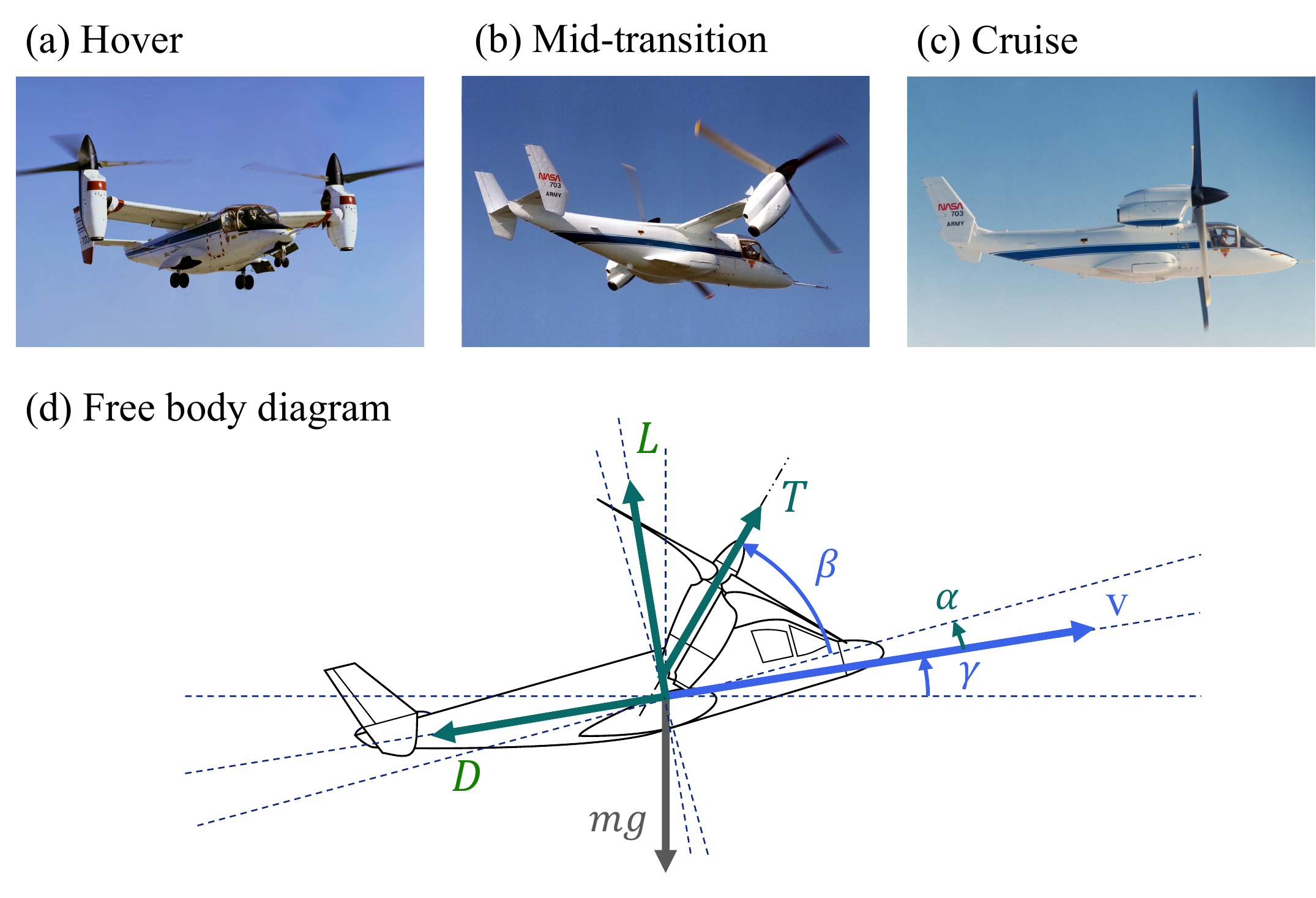}
\caption{(a)-(c) NASA-Army-Bell XV-15 tiltrotor aircraft in various rotor configurations (Source: NASA). (d) Free body diagram of XV-15 for its longitudinal dynamics in \eqref{eq:xv15-eom}.}
\label{fig:XV15}
\end{figure}

\begin{figure*}[t!]
\centering
\includegraphics[width=\textwidth]{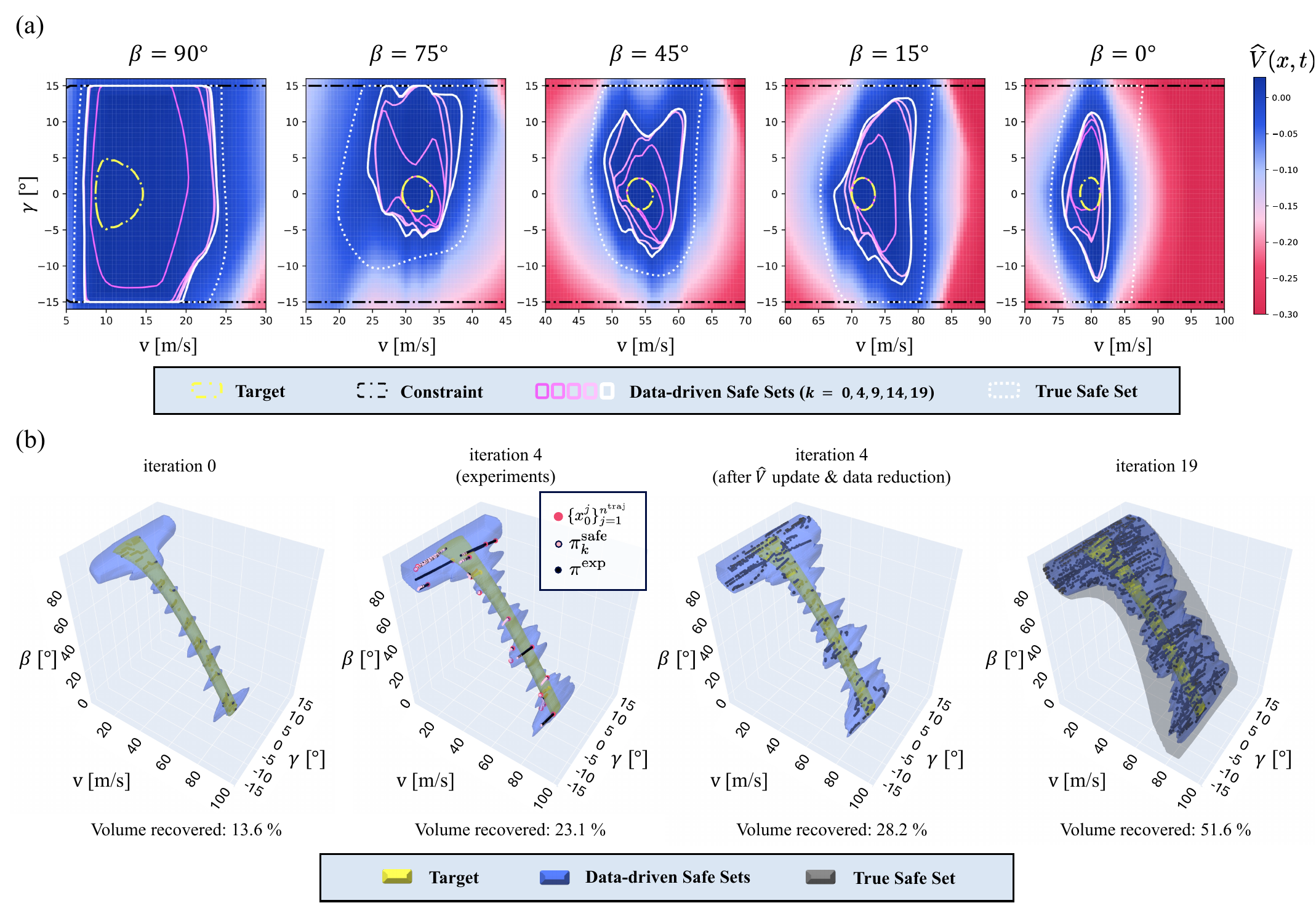}
\vspace{-1em}
\caption{XV-15 safe set expansion conducted for 20 iterations under Algorithm \ref{alg:experiment-design}, during its flight mode transition from near-hover ($\beta=90^\circ$) to cruise ($\beta=0^\circ$). (a) Data-driven safe sets at various iterations, shown as 2D slices in $\airspeed$-$\gamma$ at various tilt angles $\beta$. (b) The safe sets shown in 3D state space at iterations $k=0,4,19$. At each iteration, the trajectories are collected under the data-driven safety-filtered exploration policy, ensuring that they stay within the previous iteration's safe set. After the safe set is updated based on new data, data reduction is conducted to prune data irrelevant for safe set computation.}
\label{fig:xv15-result}
\end{figure*}

Recent progress in electric propulsion technologies has enabled various novel vertical take-off and landing (VTOL) aircraft designs to emerge. However, verifying the flight envelope of these vehicles is challenging due to their flight mode transitions between hover and cruise. During these transitions, the vehicle faces a precarious balance between rotor propulsion and aerodynamic lift, alongside highly uncertain aerodynamic interactions between the rotors and the airframe, significantly increasing the risk of loss of control \cite{belcastro2017aircraft}. This vulnerability becomes particularly apparent during flight tests in product development, where experiments to expand the vehicle's flight envelope must be conducted without an accurate model available a priori. This motivates applying our framework in \cref{sec:experiment_design} to VTOL vehicle flight envelope verification while treating the underlying dynamics as uncertain. Previously, HJ reachability was adopted in \cite{bayen2007aircraft, lombaerts2013safe, hsu2024towards} to verify the flight envelope of aircraft using explicit dynamics models. 

We consider the XV-15 \cite{maisel2000history}, a tiltrotor vehicle (Fig. \ref{fig:XV15}), transitioning from vertical to forward flight. Its mathematical model in \cite{ferguson1988mathematical} is used for simulation and validating our results. We construct the $\reachavoidbrt$ BRT from which safe recovery to a near-hover trim condition is guaranteed. Our algorithm gradually expands the BRT while ensuring safety during experiments for new data collection, analogous to the procedure of flight envelope expansion in flight tests.

\textit{Implementation Details.} We consider the reduced-order longitudinal dynamics of the vehicle detailed in \cite{ferguson1988mathematical, lombaerts2022trim}. The state consists of $x\!=\![\airspeed, \gamma, \beta]^\top$, where $\airspeed$ is the airspeed, $\gamma$ is the flight path angle, and $\beta$ is the rotor tilt angle ($\beta\!=\!0^\circ$ in cruise configuration), and the control input is $u\!=\![T, \alpha, \delta]$, where $T$ is the rotor thrust, $\alpha$ is the angle of attack, and $\delta$ is the rotor tilt angle rate. The equation of motion is given as
\begin{equation}
\resizebox{0.89\hsize}{!}{$\displaystyle
\begin{aligned}
{\left[\begin{array}{c}\!
\dot{\airspeed}\!\! \\
\!\dot{\gamma}
\!\!\end{array}\right] }\!=\!-g \!\left[\begin{array}{l}
\text{\footnotesize $\sin \gamma$}\!\\
\!\frac{\cos \gamma}{\airspeed}\!\!
\end{array}\right] \!+\!\frac{1}{m}\!\left(\left[\begin{array}{c}
\!\text{\footnotesize $\cos \left(\alpha + \beta\right)$}\!\!\\
\!\frac{\sin \left(\alpha + \beta\right)}{\airspeed}
\!\!\end{array}\right]\!T\!+\!\left[\begin{array}{c}
\!\!\text{\footnotesize $-D(\airspeed, \alpha, \beta)$}\!\!\! \\
\!\!\frac{L(\airspeed, \alpha, \beta)}{\airspeed}\!\!\!
\end{array}\right]\right),
\label{eq:xv15-eom}
\end{aligned}
$}
\end{equation}
where $L$ and $D$ are the lift and drag forces that vary with respect to $\airspeed, \alpha, \beta$, which constitutes the main nonlinearity and uncertainty source in the dynamics. For the true safe set computation, since the true Hamiltonian involves a nonconvex optimization, we approximately solve it by evaluating it over a set of discretized control inputs and then taking the maximum. For the DDH, we treat {\small $\dot{\beta}\!=\!\delta$} as a known portion of the dynamics and apply the technique in \cref{subsec:additional}. We use the hyperrectangle DDH based on a sensitivity matrix whose values are estimated from the flight data in \cite{ferguson1988mathematical}.

We consider $\gamma\!\in\![-15^\circ, 15^\circ]$ as our primary safety constraint, restraining the vehicle from extreme vertical speed, and an additional lower-bound on airspeed at each tilt angle varying from 5m/s (near-hover) to 50m/s (cruise). The initial safe set (i.e. the target set) is designed based on the LQR funnel around the transition trim corridor of the vehicle, detailed in \cite{lombaerts2022trim}. If the vehicle can reach this trim corridor, it can successfully recover to the near-hover trim state. The associated LQR control at each trim point on the corridor serves as the backup controller $\pibackup$ that maintains the target set forward invariant. For the exploration policy $\piexploration$, we apply a control input with Gaussian noise centered around a mean value randomly selected for each experiment. The mean bias term incentivizes trajectories to explore various directions. We use $\ntraj\!=\!50$ for $k\!=\!0$, and $\ntraj\!=\!20$ for subsequent iterations, with each trajectory length $\text{T}\!=\!1$s and sampling time $\Delta t\!=\!0.01$s.

The DDH solution and the numerical solver of the DDH-based HJ-VI are implemented in PyTorch, based on the level set method in \cite{mitchell2002application}. Using GPU for parallel computation, each iteration takes around 40 minutes on Nvidia RTX 4090, for the state grid with size $(191 \times 50 \times 101)$, time horizon $t=1$s, and the number of data points $N$ around 5,000.

\textit{Results.} The data-driven safe sets resulting from conducting the experiments under Algorithm \ref{alg:experiment-design} are visualized in Fig. \ref{fig:xv15-result}. After 20 iterations, we recover 51.6\% of the volume of the true safe set within the verified data-driven safe set, and each data-driven safe set is a successful inner-approximation. We also observe that the safe exploration policy designed in \eqref{eq:safe_exploration_policy} guarantees the safety of all experiments. Finally, by applying the data reduction technique, we reduced 43,000 data points sampled from trajectories to 4,294 points in the final result.

\section{Conclusion \& Future Work}
\label{sec:conclusion}

We propose a direct data-driven framework for constructing safe sets from trajectory data. This framework produces safe sets that are guaranteed to be subsets of the true safe set while only requiring knowledge of a Lipschitz constant of the dynamics. We also present an approach to iteratively expand these safe sets while maintaining safety.
The core of this framework and our main contribution is the data-driven Hamiltonian (DDH), a data-driven lower bound of the Hamiltonian used in Hamilton-Jacobi (HJ) reachability analysis. Without needing an explicit dynamics model, the DDH provides a new paradigm for how model-based analysis and prior knowledge can be integrated with data-driven approaches in order to ensure safety. 

Our proposed approach has several limitations. The numerical methods used in this work for the reachability computation suffer from the curse of dimensionality, and our DDH scales linearly with the number of data points, which can lead to the problem being intractable in higher dimensions or on large datasets. 
While we mitigate this through parallelization on GPUs and data reduction, further studies are required to improve the scalability of the computation.
The DDH also relies on knowledge of a Lipschitz constant, whose tight value may be challenging to estimate in advance.
Finally, our switching filter is rudimentary and may have issues in practice like chattering. In future work, we plan to investigate the effects of the design choices in the iterative safe set expansion algorithm, explore the use of more advanced safety filters, and extend the DDH to other control problems.

\section*{Acknowledgements}
We thank Shaun Mcwherter and Thomas Lombaerts at NASA for the insightful discussions.

\section*{Appendix: Proof of Theorem \ref{thm:main}}
For conciseness, we only conduct the proof for the $\brt$ BRT value functions. Proof for the $\reachavoidbrt$ value function can be done similarly. We consider three dynamical systems and their corresponding BRT value functions: 
\begin{enumerate}[leftmargin=1.25em, labelindent=\parindent, listparindent=\parindent, labelwidth=0em]
    \item The original dynamics $\traj(\cdot)$ in \eqref{eq:dynsys} and \eqref{eq:dynsys_cbf}, and its value function $V(x, t)$ in \eqref{eq:brt-value}.
    \item The original dynamics but whose control is confined to the ones in the dataset,
    $\dot{\tildetraj}(s)\!=\!f(\tildetraj(s), \tildectrl(s)), \tildetraj(-t) = x,$
with $\tildectrl(s)\!\in\!\tildeU \!:=\!\{u^i\}_{i=1}^{N}$. The BRT value function is given as $\tildeV(x, t) = \sup_{\tildectrl(\cdot)} \min_{s \in [-t, 0]} \unsafefcn(\tildetraj(s))$.
\item A fictitious dynamics that captures a differential Stackelberg game between the leader $\ctrlv$ and the follower $\ctrlw$, whose trajectory is defined as
\begin{equation}
\label{eq:dynsys_hat}
    \dot{\hattraj}(s) = \ctrlv(s) + \ctrlw(s), \quad \hattraj(-t) = x.
\vspace{-0.5em}
\end{equation}
The leader's action is confined by $\ctrlv(s)\!\in\!\ctrlsetv\!:=\!\{v_i\}_{i=1}^{N}$, and the follower's action is confined based on the current state and the leader's action, $\ctrlw(s) \in W(\hattraj(s), \ctrlv(s))$. The follower's action set is defined as the uncertainty sets for DDH in \eqref{eq:uncertainty-set}, i.e.,
$W(x, v_j) \equiv \uncertainset(x; x_j)$, when $\ctrlv(s) = v_j \in \ctrlsetv$. In other words, the leader selects the velocity in the dataset, and the follower selects the uncertainty vector within the uncertainty set associated with each data.

The BRT value function of this game is defined as 
\vspace{-0.5em}
\begin{equation}
\label{eq:brt-value-hat}
\hatV(x, t) := \inf_{\xi_w} \sup_{\ctrlv(\cdot)} \min_{s \in [-t, 0]} \unsafefcn(\hattraj(t)),
\vspace{-0.5em}
\end{equation}
where $\xi_w$ is the follower's non-anticipative strategy \cite{evans_hj}, in response to the leader, a mapping from $\ctrlv(\cdot)$ to $\ctrlw(\cdot)$.
\vspace{0.5em}
\end{enumerate}

\noindent Theorem \ref{thm:main} is proved by showing that a) $\hatV$ in \eqref{eq:brt-value-hat} is the unique viscosity solution of the HJ-VI in \eqref{eq:hj_vi_data} (Theorem \ref{thm:viscosity}), and b) $\hatV(x, t) \le \tildeV(x, t) \le V(x, t)$ (Theorem \ref{thm:comparison}). {\hfill $\square$}

\begin{theorem}[Viscosity solution theorem]
\label{thm:viscosity}
$\hatV$ defined in \eqref{eq:brt-value-hat} is the unique viscosity solution to the HJ-VI in \eqref{eq:hj_vi_data}.
\end{theorem}

\begin{proof}
We first see that the Hamiltonian of $\hatV$ in \eqref{eq:brt-value-hat} is indeed the DDH:
\vspace{-0.5em}
\begin{align*}
    \max_{v \in \ctrlsetv}\! & \min_{w \in W(x, v)} p^\top (v + w) = \!\! \max_{i \in \{1, \cdots, N\}} \! \min_{w \in W(x, v_i)} p^\top \! (v_i + w) \\
    & \quad = \!\!\max_{i \in \{1, \cdots, N\}} \!\min_{\widehat{v_i} \in v_i \oplus \uncertainset(x; x_i)} p^\top \widehat{v_i} = \ddh(x, p),
\vspace{-0.5em}
\end{align*}    
where $\widehat{v_i}\!=\!v_i\!+\!w$. The rest of the proof can be adopted from the viscosity solution theorem for differential games in \cite[Thm. 4.1]{evans_hj}, as similarly done in \cite{choi2021robust, fisac2015reach} for the HJ-VIs. The noticeable differences in our assumptions from those of \cite[Thm. 4.1]{evans_hj} are 1) $\ctrlsetv$ is not a compact set in our case, and 2) the follower's action space $W(x, v)$ is conditioned on the leader's action $v$. This requires the adoption of \cite[Lemma 4.3]{evans_hj}, used for the proof of \cite[Thm. 4.1]{evans_hj}, to our settings as Lemma \ref{lemma:value-to-ham} below. The lemma is the crucial step that translates the $\inf$-$\sup$ leader and follower objectives in $\hatV$ to $\max$-$\min$ objectives in $\ddh$. The uniqueness follows from \cite[Thm. 4.2]{barron1989bellman}.
\end{proof}

\begin{lemma}
\label{lemma:value-to-ham}
    (Adoption of \cite[Lemma 4.3]{evans_hj}) For $\phi (x, t) \in C^1$,

\noindent (a) If $\exists \theta > 0$, $\exists(x_0, t_0) \in \R^n \times \R_{<0}$ such that
\vspace{-0.5em}
\begin{equation}
\max_{v \in \ctrlsetv} \min_{w \in W(x, v)} D_t \phi(x, t) + D_x \phi(x, t)^\top (v + w) \le -\theta,    
\vspace{-0.5em}
\label{eq:lemma-cond-a}
\end{equation}
then there exists a small enough $\delta >0$, and the follower's non-anticipative strategy $\xi_w$ such that for all $\ctrlv(\cdot)$,
\vspace{-0.5em}
\begin{equation}
    \phi(\hattraj(t_0 + \delta), t_0 + \delta) - \phi(x_0, t_0) \le -\frac{\theta}{2}\delta.
\label{eq:lemma-result-a}
\vspace{-0.5em}
\end{equation}

\noindent (b) If $\exists \theta > 0$, $\exists(x_0, t_0)  \in \R^n \times \R_{<0}$ such that
\vspace{-0.5em}
\[\max_{v \in \ctrlsetv} \min_{w \in W(x, v)} D_t \phi(x, t) + D_x \phi(x, t)^\top (v + w) \ge \theta,
\vspace{-0.5em}
\]
there exists a small enough $\delta >0$, such that for all follower's non-anticipative strategy $\xi_w$, there exists the leader's control signal $\ctrlv(\cdot)$ such that
$
    \phi(\hattraj(t_0 + \delta), t_0 + \delta) - \phi(x_0, t_0) \ge \frac{\theta}{2}\delta.
$
\end{lemma}

\begin{proof}
\noindent Proof of (a): Due to condition \eqref{eq:lemma-cond-a}, for each $v_i \in \ctrlsetv$, there exists $w = \widehat{w}(v_i)$ such that
\vspace{-0.5em}
\begin{equation*}
    D_t \phi(x_0, t_0) + D_x \phi(x_0, t_0)^\top (v_i + \widehat{w}(v_i)) \le -\theta.
\vspace{-0.5em}
\end{equation*}
Due to the $C^1$ property of $\phi$, we have
\vspace{-0.5em}
\begin{equation}
    D_t \phi(\hattraj(s), s) + D_x \phi(\hattraj(s), s)^\top (v_i + \widehat{w}(v_i)) \le -\frac{\theta}{2},
    \label{eq:lemma-proof1}
\vspace{-0.5em}
\end{equation}
for small enough $\delta > 0$, and $s \in [t_0, t_0 + \delta]$. Consider the follower's disturbance strategy $\xi_w$ that satisfies $\xi_w[\ctrlv](s) = \widehat{w}(\ctrlv(s))$ for $s \in [t_0, t_0 + \delta]$. From \eqref{eq:lemma-proof1}, we have
\vspace{-0.5em}
\begin{equation*}
    D_t \phi(\hattraj(s), s) + D_x \phi(\hattraj(s), s)^\top (\ctrlv(s) + \xi_w[\ctrlv](s)) \le -\frac{\theta}{2},
    \vspace{-0.5em}
\end{equation*}
$\forall s \!\in \![t_0, t_0\!+\!\delta]$. Integration from $t_0$ to $t_0\!+\!\delta$ results in \eqref{eq:lemma-result-a}. Proof of (b) can be adopted directly from \cite{evans_hj}.    
\end{proof}

\begin{theorem}
\label{thm:comparison}
$\hatV(x, t) \le \tildeV(x, t) \le V(x, t)$ $\forall x\!\in\!\R^n$, $t\!\ge\! 0$.
\end{theorem}

\begin{proof}
The second inequality is trivial by observing that $\tildeU \subset \uset$.
We consider a non-anticipative strategy of the follower, $\widetilde{\xi}_{w}$, selecting its action as:
\vspace{-0.5em}
\begin{equation}
    \widetilde{\xi}_{w}[\ctrlv(s)] \equiv f(\hattraj^{\widetilde{\xi}_{w}}(s), \hatctrl(s)) - \ctrlv(s),
\vspace{-0.5em}
\end{equation}
where $\hatctrl(s)=u_j$ for $\ctrlv(s)=v_j \in \ctrlsetv$ at each $s\in[-t, 0]$, and $\hattraj^{\widetilde{\xi}_{w}}(s)$ solves \eqref{eq:dynsys_hat} with $\ctrlw(s) \equiv \widetilde{\xi}_{w}[\ctrlv(s)]$.
Note that this strategy is non-anticipative since we don't use any information of the follower's action in $\tau \in (s, 0]$. More importantly, it is a \textit{feasible} strategy because for $\ctrlv(s) = v_j$,
\vspace{-0.5em}
\begin{equation*}
     f(\hattraj^{\widetilde{\xi}_{w}}(s), u_j) - v_j \in \uncertainset(\hattraj^{\widetilde{\xi}_{w}}(s); x_j) = W(\hattraj^{\widetilde{\xi}_{w}}(s), \ctrlv(s)).
\vspace{-0.5em}
\end{equation*}
Under this follower strategy $\xi_w^*$, \eqref{eq:dynsys_hat} becomes
\vspace{-0.5em}
\[
\dot{\hattraj}^{\widetilde{\xi}_{w}}(s) = \ctrlv(s) + \widetilde{\xi}_{w}[\ctrlv(s)] = f(\hattraj^{\widetilde{\xi}_{w}}(s), \hatctrl(s)).
\vspace{-0.5em}
\]
Intuitively, the follower always selects its action so that the fictitious dynamics \eqref{eq:dynsys_hat} become identical to the original dynamics. As a result, We can see that with $\tildectrl \equiv \hatctrl$, \vspace{-0.5em}
\begin{equation*}
\sup_{\ctrlv(\cdot)} \min_{s \in [-t, 0]} \unsafefcn(\hattraj^{\widetilde{\xi}_{w}}(t)) \equiv \tildeV(x, t),
\vspace{-0.5em}
\end{equation*}
Since $\hatV$ minimizes the identical cost function over all possible follower strategies, we get $\hatV(x, t) \le \tildeV(x, t)$.
\end{proof}


\balance


\bibliographystyle{IEEEtran} 
\bibliography{IEEEabrv,references}

\end{document}